\documentclass[draftcls,onecolumn]{IEEEtran}

\usepackage{cite}
\usepackage[pdftex]{graphicx}
\usepackage{amsmath,amssymb}
\usepackage{arydshln}
\usepackage{program}
\usepackage{algpseudocode}
\usepackage{algorithm}
\usepackage{dsfont}

\usepackage{url}
\usepackage{comment}

\newtheorem{lemma}{Lemma}
\newtheorem{thm}{Theorem}

\newtheorem{ass}{Assumption}

\providecommand{\abs}[1]{\left\lvert#1 \right\rvert}
\providecommand{\norm}[1]{\left\lVert#1 \right\rVert}
\providecommand{\argmin}[1]{\underset{#1}{\operatorname{argmin}}}
\providecommand{\proj}[1]{ \Pi_{\mathcal{W}}\left[#1\right]}
\providecommand{\grad}[1]{#1}

\providecommand{\vc}[1]{\boldsymbol{#1}}

\begin{document}
%
\title{Distributed Strongly Convex Optimization}


\author{\IEEEauthorblockN{Konstantinos I. Tsianos}\\
\IEEEauthorblockA{Department of Electrical and Computer Engineering\\
McGill University\\
Montreal, Quebec H3A 0E9\\
Email: konstantinos.tsianos@gmail.com}\\
and \\
\IEEEauthorblockN{Michael G. Rabbat} \\
\IEEEauthorblockA{Department of Electrical and Computer Engineering\\
McGill University\\
Montreal, Quebec H3A 0E9\\
Email: michael.rabbat@mcgill.ca}}

\maketitle

\begin{abstract}
\boldmath
A lot of effort has been invested into characterizing the convergence rates of gradient based algorithms for non-linear convex optimization. Recently, motivated by large datasets and problems in machine learning, the interest has shifted towards distributed optimization. In this work we present a distributed algorithm for strongly convex constrained optimization. Each node in a network of $n$ computers converges to the optimum of a strongly convex, $L$-Lipchitz continuous, separable objective at a rate $O\left( \frac{\log{(\sqrt{n} T )}}{T}\right)$ where $T$ is the number of iterations. This rate is achieved in the online setting where the data is revealed one at a time to the nodes, and in the batch setting where each node has access to its full local dataset from the start. The same convergence rate is achieved in expectation when the subgradients used at each node are corrupted with additive zero-mean noise.
\end{abstract}


%
\IEEEpeerreviewmaketitle

\section{Introduction} 
In this work we focus on solving optimization problems of the form
\begin{align} \label{eq:general_optimization_problem}
\underset{w \in \mathcal{W}}{\operatorname{minimize}}\ F(w) = \frac{1}{T} \sum_{t=1}^T f^t(w)
\end{align}
where each function $f^1(w), f^2(w), \dots,$ is convex over a convex set $\mathcal{W} \subseteq \mathds{R}^d$. This formulation applies widely in machine learning scenarios, where $f^t(w)$ measures the loss of model $w$ with respect to data point $t$, and $F(w)$ is the average loss over $T$ data points. In particular, we are interested in the behavior of online distributed optimization algorithms for this sort of problem as the number of data points $T$ tends to infinity. We describe a distributed algorithm which, for strongly convex functions $f^t$, converges at a rate $O\left( \frac{\log(\sqrt{n} T)}{T}\right)$. To the best of our knowledge this is the first distributed algorithm to achieve this converge rate for constrained optimization without relying on smoothness assumptions on the objective or non-trivial communication mechanisms between the nodes. The result is true both in the online and the batch optimization setting.

When faced with a non-linear convex optimization problem, gradient-based methods can be applied to find the solution. The behavior of these algorithms is well-understood in the single-processor (centralized) setting. Under the assumption that the objective is $L$-Lipschitz continuous, projected gradient descent-type algorithms converge at a rate $O(\frac{1}{\sqrt{T}})$~\cite{zinkOnlineConvexOpt,nesterovDualOpt}. This rate is achieved both in an online setting where the $f^t$'s are revealed to the algorithm sequentially and in the batch setting where all $f^t$ are known in advance. If the cost functions are also strongly convex then gradient algorithms can achieve linear rates, $O\left(\frac{1}{T}\right)$, in the batch setting \cite{largeScaleSGD} and nearly-linear rates, $O\left(\frac{\log(T)}{T}\right)$, in the online setting \cite{logRegRepeatedGames}. Under additional smoothness assumptions, such as Lipschitz continuous gradients, the same rate of convergence can also be achieved by second order methods in the online setting~\cite{LogRegretOCOot,AdaptiveOGD}, while accelerated methods can achieve a quadratic rate in the batch setting; see~\cite{TsengAccelerated} and references therein.

The aim of this work is to extend the aforementioned results to the distributed setting where a network of processors jointly optimize a similar objective. Assuming the network is arranged as an expander graph with constant spectral gap, for general convex cost functions that are only $L$-Lipschitz continuous, the rate at which existing algorithms on a network of $n$ processors will all reach the optimum value is $O(\frac{\log(T \sqrt{n})}{\sqrt{T}})$, i.e., similar to the optimal single processor algorithms up to a logarithmic factor~\cite{dualAveraging,ContrainedDistrOpt}. This is true both in a batch setting and in an online setting, even when the gradients are corrupted by noise. The technique proposed in~\cite{DekelMiniBatches} makes use of mini-batches to obtain asymptotic rates $O\left( \frac{\sqrt{\log{(n)}}}{\sqrt{n T}}\right)$ for online optimization of smooth cost functions that have Lipschitz continuous gradients corrupted by bounded variance noise, and $O\left( \frac{1}{nT}\right)$ for smooth strongly convex functions. However, this technique requires that each node exchange messages with every other node at the end of each iteration. Finally, if the objective function is strongly convex and three times differentiable, a distributed version of Nesterov's accelerated method~\cite{FastDistributedGradMethods} achieves a rate of $O\left(\frac{\log(T)}{T}\right)$ for unconstrained problems in the batch setting, but the dependence on $n$ is not characterized.

The algorithm presented in this paper achieves a rate $O\left( \frac{\log{(\sqrt{n} T )}}{T}\right)$ for strongly convex functions. Our formulation allows for convex constraints in the problem and assumes the objective function is Lipschitz continuous and strongly convex; no higher-order smoothness assumptions are made. Our algorithm works in both the online and batch setting and it scales nearly-linearly in number of iterations for network topologies with fast information diffusion. In addition, at each iteration nodes are only required to exchange messages with a subset of other nodes in the network (their neighbors).

The rest of the paper is organized as follows. Section~\ref{sec:into_convex_opt} introduces notation and formalizes the problem. Section~\ref{sec:algorithm} describes the proposed algorithm and states our main results. These results are proven in Section~\ref{sec:analysis}, and Section~\ref{sec:stochastic_opt} extends the analysis to the case where gradients are noisy. Section~\ref{sec:experiments} presents the results of numerical experiments illustrating the performance of the algorithm, and the paper concludes in Section~\ref{sec:future}.

\section{Online Convex Optimization}
\label{sec:into_convex_opt}
Consider the problem of minimizing a convex function $F(w)$ over a convex set $\mathcal{W} \subseteq \mathds{R}^d$. Of particular interest is the setting where the algorithm sequentially receives noisy samples of the (sub)gradients of $F(w)$. This setting arises in online loss minimization for machine learning when the data arrives as a steam and the (sub)gradient is evaluated using an individual data point at each step~\cite{zinkOnlineConvexOpt}. Suppose the $t$th data point $x(t) \in \mathcal{X} \subseteq \mathds{R}^d$ is drawn i.i.d.~from an unknown distribution $\mathcal{D}$, and let $f^t(w) = f(w, x(t))$ denote the loss of this data point with respect to a particular model $w$. In this setting one would like to find the model $w$ that minimizes the expected loss $\mathds{E}_{\mathcal{D}}[f(w,x)]$, possibly with the constraint that $w$ be restricted to a model space $\mathcal{W}$. Clearly, as $T \rightarrow \infty$, the objective $F(w) = \frac{1}{T} \sum_{t=1}^T f^t(w) \rightarrow \mathds{E}_{\mathcal{D}}[f(w,x)]$, and so if the data stream is finite this motivates minimizing the empirical loss $F(w)$.

An online convex optimization algorithm observes a data stream $x(1), x(2), \dots$, and sequentially chooses a sequence of models $w(1), w(2), \dots$, after each observation. Upon choosing $w(t)$, the algorithm receives a subgradient $g(t) \in \partial f^t(w(t))$. The goal is for the sequence $w(1), w(2), \dots$ to converge to a minimizer $w^*$ of $F(w)$.


The performance of  an online optimization algorithm is measured in terms of the \textit{regret}:
\begin{align}
R(T) = \sum_{t=1}^T f^t(w(t))  - \min_{w \in \mathcal{W}} \sum_{t=1}^T f^t(w).
\end{align}
The regret measures the gap between the cost accumulated by the online optimization algorithm over $T$ steps and that of a model chosen to simultaneously minimize the total regret over all $T$ cost terms. If the costs $f^t$ are allowed to be arbitrary convex functions then it can be shown that the best achievable rate for any online optimization algorithm is $\frac{R(T)}{T} = \Omega(\frac{1}{\sqrt{T}})$, and this bound is also achievable~\cite{zinkOnlineConvexOpt}. The rate can be significantly improved if the cost functions has more favourable properties. 

\subsection{Assumptions}

\begin{ass} \label{ass:stronglyConvex}
We assume for the rest of the paper that each cost function $f^t(w) = f(w, x(t))$ is $\sigma$-strongly convex for all $x(t) \in \mathcal{X}$; i.e., there is a $\sigma > 0$ such that for all $\theta \in [0,1]$ and all $u,w \in \mathcal{W}$
\begin{align}
f^t(\theta u +& (1- \theta) w )  \leq  \notag  \\
&\theta f^t(u) + (1-\theta) f^t(w) - \frac{\sigma}{2} \theta (1 - \theta) \norm{u - w}^2.
\end{align}
\end{ass}

If each $f^t(w)$ is $\sigma$-strongly convex, it follows that $F(w)$ is also $\sigma$-strongly convex. Moreover, if $F(w)$ is strongly convex then it is also strictly convex, and so $F(w)$ has a unique minimizer which we denote by $w^*$.

\begin{ass} \label{ass:boundedGradients}
We also assume that the subgradients $\grad{g}(t)$ of each cost function $f^t$ are bounded by a known constant $L > 0$; i.e., $\norm{\grad{g}(t)} \leq L$ where $\norm{\cdot}$ is the ($\ell_2$) Euclidean norm.
\end{ass}

\subsection{Example: Training a Classifier}

For a specific example of this setup, consider the problem of training an SVM classifier using a hinge-loss with $\ell_2$ regularization \cite{logRegRepeatedGames}. In this case, the data stream consists of pairs $\{x(t), y(t)\}$ such that $x(t) \in \mathcal{X}$ and $y(t) \in \{-1, +1\}$. The goal is to minimize the misclassification error as measured by the $\ell_2$-regularized hinge loss. Formally, we wish to find the $w^* \in \mathcal{W} \subseteq \mathds{R}^d$ that solves
\begin{align} \label{eq:svm}
\operatorname{minimize}_{w \in \mathcal{W}} \frac{\sigma}{2} \norm{w}^2 + \frac{1}{m} \sum_{t=1}^{m} \max\{0, 1 - y(t) \langle w, x(t)\rangle\}
\end{align}
which is $\sigma$-strongly convex\footnote{Although the hinge loss itself is not strongly convex, adding a strongly convex regularizer makes the overall cost function strongly convex.}. For these types of problems, using a single-processor stochastic gradient descent algorithm, one can achieve $\frac{R(T)}{T} = O(\frac{\log{T}}{T})$\cite{logRegRepeatedGames} or $\frac{R(T)}{T} = O(\frac{1}{T})$ \cite{StochOptStrognlyConvex} by using different update schemes. 

\subsection{Distributed Online Convex Optimization}

In this paper, we are interested in solving online convex optimization problems with a network of computers. The computers are organized as a network $G=(V,E)$ with $\abs{V} = n$ nodes, and messages are only exchanged between nodes connected with an edge in $E$. 

\begin{ass} \label{ass:connected}
In this work we assume that $G$ is connected and undirected. 
\end{ass}

Each node $i$ receives a stream of data $x_i(1), x_i(2), \dots$, similar to the serial case, and the nodes must collaborate to minimize the network-wide objective 
\begin{align}
F(w) = \frac{1}{nT} \sum_{t=1}^T \sum_{i=1}^n f_i^t(w),
\end{align}
where $f_i^t(w) = f(w, x_i(t))$ is the cost incurred at processor $i$ at time $t$. In the distributed setting, the definition of regret is naturally extended to
\begin{align}
R(T) = \sum_{t=1}^T \sum_{i=1}^n f(w_i(t), x_i(t))  - \min_{w \in \mathcal{W}} \sum_{t=1}^T \sum_{i=1}^n f(w, x_i(t)).
\end{align}

For general convex cost functions, the distributed algorithm proposed in~\cite{dualAveraging} has been proven to have an average regret that decreases at a rate $\sqrt{T}$, similar to the serial case, and this result holds even when the algorithm receives noisy, unbiased, observations of the true subgradients at each step. In the next section, we present a distributed algorithm that achieves a nearly-linear rate of decrease of the average regret (up to a logarithmic factor) when the cost functions are strongly convex.

\section{Algorithm} \label{sec:algorithm}

Nodes must collaborate to solve the distributed online convex optimization problem described in the previous section. To that end, the network is endowed with a $n \times n$ consensus matrix $P$ which respects the structure of $G$, in the sense that $[P]_{ji} = 0$ if $(i,j) \notin E$. We assume that $P$ is doubly stochastic, although generalizations to the case where $P$ is row stochastic or column stochastic (but not both) are also possible \cite{TsianosCDC2012,TsianosACC2012}.

A detailed description of the proposed algorithm, \emph{distributed online gradient descent} (DOGD), is given in Algorithm~\ref{alg:dogd}. In the algorithm, each node performs a total of $T$ updates. One update involves processing a single data point $x_i(t)$ at each processor. The updates are performed over $k$ rounds, and $T_s$ updates are performed in round $s \leq k $. The main steps within each round (lines 9--11) involve updating an accumulated gradient variable, $z_i^k(t)$, by simultaneously incorporating the information received from neighboring nodes and taking a local gradient-descent like step. The accumulated gradient is projected onto the constraint set to obtain $w_i^k(t)$, where
\begin{align}
\proj{z} = \argmin{w \in \mathcal{W}} \norm{w - z}
\end{align}
denotes the Euclidean projection of $z$ onto $\mathcal{W}$, and then this projected value is merged into a running average $\hat{w}_i(r)$. The step size parameter $a_k$ remains constant within each round, and the step size is reduced by half at the end of each round. The number of updates per round doubles from one round to the next.

Note that the algorithm proposed here differs from the distributed dual averaging algorithm described in~\cite{dualAveraging}, where a proximal projection is used rather than the Euclidean projection. Also, in contrast to the distributed subgradient algorithms described in~\cite{distrStochSubgrOpt}, DOGD maintains an accumulated gradient variable in $z_i^k(t+1)$ which is updated using $\{z_j^k(t)\}$ as opposed to the primal feasible variables $\{w_j^k(t)\}$. Finally, key to achieving fast convergence is the exponential decrease of the learning rate after performing an exponentially increasing number of gradient steps together with a proper initialization of the learning rate. 

The next section provides theoretical guarantees on the performance of DOGD.

\begin{algorithm}[t]
\caption{DOGD}
\label{alg:dogd}
\begin{algorithmic}[1]
\State $\operatorname{Initialize:} T_1 =  \left\lceil \frac{2 }{\sigma} \right\rceil, a_1 = 1, k=1, z_i^1(1) = w_i^1(1) = 0$
\State
\While{$\sum_{s=1}^k T_s \leq T$} \Comment{Each node $i$ repeats}
\For{$t=1$ to $T_k$} 
\State Send/receive $z_i^k(t)$ and $z_j^k(t)$ to/from neighbors
\State Obtain next subgradient $\grad{g}_i(t) \in \partial_w f_i^t(w_i^k(t))$
\State $z_i^k(t+1) = \sum_{j=1}^n p_{ij} z_j^k(t) - a_k \grad{g}_i(t)$
\State $w_i^k(t+1) = \proj{z_i^k(t+1)}$
\EndFor
\State $w_i^{k+1}(1) = w_i^k(T_k)$
\State $z_i^{k+1}(1) = w_i^{k+1}(1)$
\State $\hat{w}_i^{k+1} = \frac{1}{T_k} \sum_{t=1}^{T_k} w_i^k(t) $
\State
\State $T_{k+1} \leftarrow 2  T_k$
\State $a_{k+1} \leftarrow  \frac{a_k}{2}$
\State $k = k + 1$
\EndWhile
\end{algorithmic}
\end{algorithm}

\section{Convergence Analysis} \label{sec:analysis}

Our main convergence result, stated below, guarantees that the average regret decreases at a rate which is nearly linear.

\begin{thm} \label{thm:dogd_convergence} 
Let Assumptions~\ref{ass:stronglyConvex}--\ref{ass:connected} hold and suppose that the consensus matrix $P$ is doubly stochastic with constant $\lambda_2$. Let $w^*$ be the minimizer of $F(w)$. Then the sequence $\{\hat{w}_i^k\}$ produced by nodes running DOGD to minimize $F(w)$ obeys
\begin{align}
F(\hat{w}_i^{k+1}) - F(w^*) = O\left( \frac{\log{(\sqrt{n} T )}}{T} \right),
\end{align}
where $k = \lfloor \log_2(T/2 + 1)\rfloor$ is the number of rounds executed during a total of $T$ gradient steps per node, and $\hat{w}_i^{k}$ is the running average maintained locally at each node.
\end{thm}

\textit{Remark 1:} We state the result for the case where $\lambda_2$ is constant. This is the case when $G$ is, e.g., a complete graph or an expander graph \cite{kRegExpanders}. For other graph topologies where $\lambda_2$ shrinks with $n$ and consensus does not converge fast, the convergence rate dependence on $n$ is going to be worse due to a factor $1 - \sqrt{\lambda_2}$ in the denominator; see the proof of Theorem~\ref{thm:dogd_convergence} below for the precise dependence on the spectral gap $1 - \sqrt{\lambda_2}$.

\textit{Remark 2:} The theorem characterizes performance of the online algorithm DOGD, where the data and cost functions $f^t_i$ are processed sequentially at each node in order to minimize an objective of the form 
\begin{align}
F(w) = \frac{1}{n} \sum_{i=1}^n \frac{1}{T} \sum_{t=1}^T f_i^t(w).
\end{align}
However, as pointed out in \cite{logRegRepeatedGames}, if the entire dataset is available in advance, we can use the same scheme to do batch minimization by effectively setting $f_i^t(w) = f_i^1(w)$, where $f_i^1(w)$ is the objective function accounting for the entire dataset available to node $i$. Thus, the same result holds immediately for a batch version of DOGD.

The remainder of this section is devoted to the proof of Theorem~\ref{thm:dogd_convergence}. Our analysis follows arguments that can be found in \cite{zinkOnlineConvexOpt, StochOptStrognlyConvex,dualAveraging} and references therein. We first state and prove some intermediate results.

\subsection{Properties of Strongly Convex Functions}

Recall the definition of $\sigma$-strong convexity given in Assumption~\ref{ass:stronglyConvex}. A direct consequence of this definition is that if $F(w)$ is $\sigma$-strongly convex then
\begin{align}
F(w) - F(w^*) \geq \frac{\sigma}{2} \norm{w - w^*}^2.
\end{align}

Strong convexity can be combined with the assumptions above to upper bound the difference $F(w)-F(w^*)$ for an arbitrary point $w \in \mathcal{W}$.

\begin{lemma} \label{lem:startingPoint}
Let $w^*$ be the minimizer of $F(w)$. For all $w \in \mathcal{W}$, we have $F(w) - F(w^*) \leq \frac{2 L^2}{\sigma}$.
\end{lemma}
\begin{proof}
For any subgradient $\grad{g}$ of $F$ at $w$, by convexity we know that $F(w) - F(w^*) \leq \langle \grad{g}, w - w^* \rangle$. It follows from Assumption~\ref{ass:boundedGradients} that $F(w) - F(w^*) \leq L \norm{w- w^*} $. Furthermore, from Assumption~\ref{ass:stronglyConvex}, we obtain that $\frac{\sigma}{2} \norm{w - w^*}^2 \leq L \norm{w - w^*}$ or $\norm{w - w^*} \leq \frac{2 L}{ \sigma}$. As a result, $F(w) - F(w^*) \leq \frac{2 L^2}{\sigma}$.
\end{proof}

\subsection{The Lazy Projection Algorithm}

The analysis of DOGD below involves showing that the average state, $\frac{1}{n}\sum_{i=1}^n w_i^k(t)$, evolves according to the so-called (single processor) \textit{lazy projection} algorithm~\cite{zinkOnlineConvexOpt}, which we discuss next. The lazy projection algorithm is an online convex optimization scheme for the serial problem discussed at the beginning of Section~\ref{sec:into_convex_opt}. A single processor sequentially chooses a new variable $w(t)$ and receives a subgradient $g(t)$ of $f(w(t), x(t))$. The algorithm chooses $w(t+1)$ by repeating the steps
\begin{align}
z(t+1) = & z(t) - a \grad{g}(t) \label{eqn:lazyProjection1} \\ 
w(t+1) = & \proj{z(t+1)} \label{eqn:lazyProjection2}.
\end{align}
By unwrapping the recursive form of \eqref{eqn:lazyProjection1}, we get
\begin{align}
z(t+1) = -a \sum_{s=1}^t g(t) + z(1). \label{eqn:lazyProjectionUnwrapped}
\end{align}

The following is a typical result for subgradient descent-style algorithms, and is useful towards eventually characterizing how the regret accumulates. Its proof can be found in the appendix of the extended version of~\cite{zinkOnlineConvexOpt}.

\begin{thm}[Zinkevich~\cite{zinkOnlineConvexOpt}] \label{thm:zink} 
Let $w(1) \in \mathcal{W}$, let $a > 0$, and set $z(1) = w(1)$. After $T$ rounds of the serial lazy projection algorithm \eqref{eqn:lazyProjection1}--\eqref{eqn:lazyProjection2}, we have
\begin{align}
\sum_{t=1}^T \langle \grad{g}(t), w(t) - w^* \rangle \leq \frac{\norm{w(1) - w^*}^2}{2 a} + \frac{T a L^2}{2}.
\end{align}
\end{thm}
Theorem \ref{thm:zink} immediately yields the same bound for the regret of lazy projection \cite{zinkOnlineConvexOpt}.

\subsection{Evolution of Network-Average Quantities in DOGD}

We turn our attention to Algorithm \ref{alg:dogd}. A standard approach to studying convergence of distributed optimization algorithms, such as DOGD, is to keep track of the discrepancy between every node's state and an average state sequence defined as
\begin{align}
\overline{z}^k(t) = \frac{1}{n} \sum_{i=1}^n z_i^k(t) \quad \text{and} \quad \overline{w}^k(t) = \proj{\overline{z}^k(t)}.
\end{align}
Observe that $\overline{z}^k(t)$ evolves in a simple recursive manner,
\begin{align}
\overline{z}^k(t+1) = &\frac{1}{n} \sum_{i=1}^n z_i^k(t+1) \\
= & \frac{1}{n} \sum_{i=1}^n \left[ \sum_{j=1}^n p_{ij} z_j^k(t) - a_k \grad{g}_i(t) \right] \\
= & \frac{1}{n} \sum_{j=1}^n z_j^k(t) \sum_{i=1}^n  p_{ij} - \frac{ a_k}{n} \sum_{i=1}^n \grad{g}_i(t) \\
= & \overline{z}(t) - \frac{ a_k}{n} \sum_{i=1}^n \grad{g}_i(t)  \label{eq:doubly_stoch_P}\\
= & -a_k \sum_{s=1}^t \frac{1}{n} \sum_{i=1}^n \grad{g}_i(s) + \frac{1}{n}\sum_{i=1}^n z_i^k(1) \label{eq:zbar}
\end{align}
where equation \eqref{eq:doubly_stoch_P} holds since $P$ is doubly stochastic. Notice (cf.~eqn.~\eqref{eqn:lazyProjectionUnwrapped}) that the states $\{\overline{z}^k(t), \overline{w}^k(t)\}$ evolve according to the lazy projection algorithm with gradients $\overline{g}(t)= \frac{1}{n}  \sum_{i=1}^n \grad{g}_i(t) $ and learning rate $a_k$. 
In the sequel, we will also use an analytic expression for $z_i^k(t)$ derived by back substituting in its recursive update equation. After some algebraic manipulation, we obtain
\begin{align}
z_i^k(t) = & - a_k \sum_{s=1}^{t-1} \sum_{j=1}^n \left[P^{t-s+1}\right]_{ij} \grad{g}_j(s-1) - a_k \grad{g}_i(t-1) \nonumber \\
& + \sum_{j=1}^n [P^t]_{ij} z_j^k(1), \label{eqn:DOGD_unwrapped}
\end{align}
and since the projection in non-expansive and $z_i^1(1) = 0, \forall i$,
\begin{align}
\norm{z_i^{k+1}(1)}= & \norm{w_i^{k+1}(1)} = \norm{w_i^k(T_k)} = \norm{\proj{z_i^k(T_k)}}   \\
\leq &\norm{z_i^{k}(T_k)} \\
 \leq & \norm{ - a_{k} \sum_{t=1}^{T_{k}-1} \sum_{i=1}^n  \left[ P^{T_k - s + 1}\right]_{ij} g_i(s-1)} \nonumber \\
& + \norm{-a_k g_i(T_k - 1)} + \sum_{j=1}^n \left[P^{T_k} \right]_{ij} \norm{z_j^k(1) }  \\
\leq & a_k T_k L + \sum_{j=1}^n \left[P^{T_k} \right]_{ij} \norm{z_j^k(1) } \\
\leq & \cdots \\
\leq & L \sum_{s=1}^{k} a_{s} T_{s}. \label{eq:znorm_bound} 
\end{align}

\subsection{Analysis of One Round of DOGD}

Next, we focus on bounding the amount of regret accumulated during the $k$th round of DOGD (lines 5--12 of Algorithm~1) during which the learning rate remains fixed at $a_k$. 
Using Assumptions~\ref{ass:stronglyConvex}, \ref{ass:boundedGradients}, and the triangle inequality we have that
\begin{align}
\sum_{t=1}^{T_k} & [ F(w_i^{k}(t)) -   F(w^*)] \notag \\
= &  \sum_{t=1}^{T_k}\left[F\left(\overline{w}^k(t) \right) - F(w^*)  + F(w_i^{k}(t)) - F\left( \overline{w}^k(t) \right) \right]  \\
\leq &   \sum_{t=1}^{T_k} \left[ F( \overline{w}^k(t) ) - F(w^*)  + L\norm{  w_i^{k}(t) - \overline{w}^k(t) } \right]  \\
\leq & \sum_{t=1}^{T_k}  \frac{1}{n} \sum_{i=1}^n  \left[ f_i(w_i^k(t)) - f_i(w^*) \right]  \notag \\
& + \sum_{t=1}^{T_k} \frac{1}{n} \sum_{i=1}^n  \left[ f_i(\overline{w}^k(t))-  f_i(w_i^k(t)) \right]  \notag \\
& + \sum_{t=1}^{T_k} L \norm{w_i^k(t) - \overline{w}^k(t)} \\
 \leq  &  \underbrace{\sum_{t=1}^{T_k} \frac{1}{n}  \sum_{i=1}^n \langle g_i(t), w_i^k(t) - w^*\rangle}_{A_1}   \notag \\
& + \sum_{t=1}^{T_k} \frac{1}{n} \sum_{i=1}^n L \norm{\overline{w}^k(t) -  w_i^k(t)}  \notag \\
& + \sum_{t=1}^{T_k} L \norm{w_i^k(t) - \overline{w}^k(t)}. 
\end{align}
For the first summand we have
\begin{align}
A_1 =  &\sum_{t=1}^{T_k} \frac{1}{n} \sum_{i=1}^n \langle g_i(t), w_i^k(t) - w^*\rangle \\
 \leq & \sum_{t=1}^{T_k} \frac{1}{n}  \sum_{i=1}^n \langle \grad{g}_i(t) , \overline{w}^k(t) - w^*\rangle  \notag \\
 & +    \sum_{t=1}^{T_k} \frac{1}{n} \sum_{i=1}^n \langle \grad{g}_i(t) , w_i^k(t) - \overline{w}^k(t)\rangle  \\
  \leq & \underbrace{\sum_{t=1}^{T_k} \frac{1}{n}  \sum_{i=1}^n \langle \grad{g}_i(t) , \overline{w}^k(t) - w^*\rangle}_{A_2}   \notag \\
 & +    \sum_{t=1}^{T_k} \frac{1}{n} \sum_{i=1}^n L \norm{w_i^k(t) - \overline{w}^k(t)}. \label{eq:optimization_error_term}
\end{align}

To bound term $A_2$ we invoke Theorem \ref{thm:zink} for the average sequences $\{\overline{w}^k(t)\}$ and $\{\overline{z}^k(t)\}$. 
\begin{align}
A_2  = & \sum_{t=1}^{T_k} \frac{1}{n} \sum_{i=1}^n \langle  \grad{g}_i(t) , \overline{w}^k(t) - w^*\rangle \label{eq:bound_starting_point} \\
 = & \sum_{t=1}^{T_k}  \Big\langle \frac{1}{n} \sum_{i=1}^n \grad{g}_i(t) , \proj{ \overline{z}^k(t)}  - w^* \Big\rangle \\
  = & \sum_{t=1}^{T_k}  \Big\langle \overline{\grad{g}}(t) , \proj{ \overline{z}^k(t)}  - w^* \Big\rangle \\
\leq & \frac{\norm{\overline{w}^k(1) - w^*}^2}{2 a_k} + \frac{{T_k} a_k \norm{\frac{1}{n} \sum_{i=1}^n \grad{g}_i(t)}^2 }{2} \\
= &  \frac{\norm{\overline{w}^k(1) - w^*}^2}{2 a_k} + \frac{{T_k} a_k L^2 }{2}.
\end{align}
Collecting now all the partial results and bounds, so far we have shown that
\begin{align}
\sum_{t=1}^{T_k}  [ F(w_i^{k}(t)) & -  F(w^*) ]\leq  \frac{\norm{\overline{w}^k(1) - w^*}^2}{2 a_k} + \frac{{T_k} a_k L^2 }{2} \notag \\
 & +\sum_{t=1}^{T_k} \frac{2}{n}  \sum_{i=1}^n L  \norm{w_i^k(t) - \overline{w}^k(t)}  \notag \\
& + \sum_{t=1}^{T_k} L\norm{w_i^k(t) - \overline{w}^k(t)} .
\end{align}
and since the projection operator is non-expansive, we have
\begin{align}
\sum_{t=1}^{T_k}  [ F(w_i^{k}(t)) & -  F(w^*) ] \leq  \frac{\norm{\overline{w}^k(1) - w^*}^2}{2 a_k} + \frac{{T_k} a_k L^2 }{2} \nonumber \\
 & +\sum_{t=1}^{T_k} \frac{2}{n}  \sum_{i=1}^n L \norm{z_i^k(t) - \overline{z}^k(t)}  \label{eqn:oneRoundBound}  \\
& + \sum_{t=1}^{T_k} L  \norm{z_i^k(t) - \overline{z}^k(t)} . \nonumber
\end{align}
The first two terms are standard for subgradient algorithms using a constant step size. The last two terms depend on the error between each node's iterate $z_i^k(t)$ and the network-wide average $\overline{z}^k(t)$, which we bound next.

\subsection{Bounding the Network Error}
What remains is to bound the term $\norm{z_i^k(t) - \overline{z}^k(t)}$ which describes an error induced by the network since the different nodes do not agree on the direction towards the optimum. By recalling that $P$ is doubly stochastic and manipulating the recursive expressions \eqref{eqn:DOGD_unwrapped} and \eqref{eq:zbar} for $z_i(t)$ and $\overline{z}^k(t)$ using arguments similar to those in~\cite{dualAveraging,TsianosACC2012}, we obtain the bound,
\begin{align}
\norm{z_i^k(t) - \overline{z}^k(t)} \leq & a_k L \sum_{s=1}^{t-1} \sum_{j=1}^n  \abs{ \frac{1}{n} \vc{1}^{T} - \left[ P^{t-s-1}\right]_{ij}}_1 + 2 a_k L \notag \\
& + \sum_{j=1}^n \abs{ \frac{1}{n} - [P^t]_{ij} } \norm{ z_j^k(1)} \\
= & a_k L \sum_{s=1}^{t-1} \norm{\frac{1}{n} \vc{1}^{T} - \left[ P^{t-s-1}\right]_{i,:}}_1 + 2 a_k L \notag \\
& + \sum_{j=1}^n \abs{ \frac{1}{n} - [P^t]_{ij} } \norm{ z_j^k(1)}.
\end{align}
The $\ell_1$ norm can be bounded using Lemma~\ref{lem:P_tv_convergence}, which is stated and proven in the Appendix, and using \eqref{eq:znorm_bound} we arrive at
\begin{align}
\norm{z_i^k(t) - \overline{z}^k(t)} \leq &2 a_k L \frac{\log{(T_k \sqrt{n})}}{1 - \sqrt{\lambda_2}} + 3 a_k L   + \frac{L \sum_{s=1}^{k-1} a_{s} T_{s} }{T_k}
\end{align}
where $\lambda_2$ is the second largest eigenvalue of $P$. Using this bound in equation~\eqref{eqn:oneRoundBound},  along with the fact that $F(w)$ is convex, we conclude that
\begin{align}
F(\hat{w}_i^{k+1}) - F(w^*)  = & F\left(\frac{1}{T_k} \sum_{t=1}^{T_k} w_i^{k}(t) \right) - F(w^*)\\
\leq & \frac{1}{T_k} \sum_{t=1}^{T_k} \left[ F( w_i^{k}(t) ) - F(w^*) \right] \\
 \leq &  \frac{\norm{\overline{w}^k(1) - w^*}^2}{2 a_k T_k} + \frac{ a_k L^2 }{2} \notag \\
& + L^2 a_k \left[ 6 \frac{\log{(T_k \sqrt{n})}}{1 - \sqrt{\lambda_2}} + 9 \right] \notag \\
& + \frac{3 L^2 \sum_{s=1}^{k-1} a_{s} T_{s} }{T_k},
\end{align}
where $\overline{w}^k(1) = \proj{\frac{1}{n} \sum_{i=1}^{n} z_i^k(1) }$.

\subsection{Analysis of DOGD over Multiple Rounds}

As our last intermediate step, we must control the learning rate and update of $T_k$ from round-to-round to ensure linear convergence of the error. From strong convexity of $F$ we have
\begin{align}
\norm{\overline{w}^k(1) - w^*}^2 \leq  2 \frac{ F(\overline{w}^k(1)) - F(w^*)}{\sigma}
\end{align}
and thus
\begin{align}
F(\hat{w}_i^{k+1}) & - F(w^*)  \leq   \frac{ F(\overline{w}^k(1)) - F(w^*)}{\sigma a_k T_k} \notag \\
&  + \frac{L^2 a_k}{2} \left[ 12 \frac{\log{(T_k \sqrt{n})}}{1 - \sqrt{\lambda_2}} + 19\right]  \notag \\
& + \frac{3 L^2 \sum_{s=1}^{k-1} a_{s} T_{s} }{T_k}\label{eq:Fwhat_bound}.
\end{align}
Now, from Theorem $3$ in \cite{zinkOnlineConvexOpt} which is a direct consequence of Theorem \ref{thm:zink} for the average sequence $\overline{w}$ viewed as a single processor lazy projection algorithm, we have that after executing $T_{k-1}$ gradient steps in round $k-1$,
\begin{align}
 F(\overline{w}^k(1)) - F(w^*) \leq \frac{\norm{\overline{w}^{k-1}(1) - w^*}^2}{2 a_{k-1} T_{k-1}} + \frac{ a_{k-1} L^2}{2}
\end{align}
and by repeatedly using strong convexity and  Theorem \ref{thm:zink} we see that
\begin{align}
 F(\overline{w}^k(1)) - F(w^*) \leq & \frac{F(\overline{w}^{k-1}(1)) - F(w^*)}{\sigma a_{k-1} T_{k-1}} + \frac{ a_{k-1} L^2}{2} \\
 \leq & \cdots \\
 \leq & \frac{F(\overline{w}^{1}(1)) - F(w^*)}{\prod_{j=0}^{k-1} (\sigma a_{k-j} T_{k-j})} \notag \\
 & + \sum_{j=1}^{k-1} \frac{a_{k-j} L^2}{2 \prod_{s=1}^{j-1} (\sigma a_{k-s} T_{k-s})} \label{eq:Fwbar_bound}.
\end{align}
Now, let us fix positive integers $b$ and $c$, and suppose we use the following rules to determine the step size and number of updates performed within each round: 
\begin{align}
a_k = &\frac{a_{k-1}}{b} = \cdots = \frac{a_1}{b^{k-1}} \\
T_k = & c T_{k-1} = \cdots = c^{k-1} T_1.
\end{align}
Combining \eqref{eq:Fwbar_bound} with \eqref{eq:Fwhat_bound} and invoking Lemma \ref{lem:startingPoint}, we have
\begin{align}
F(\hat{w}_i^{k+1}) & - F(w^*)  \leq  \frac{2 L^2}{\sigma \prod_{j=0}^{k-1} \left( \sigma a_1 T_1 \left( \frac{c}{b} \right)^{k-j-1} \right)}  \notag \\
& + \sum_{j=1}^{k-1} \frac{a_{1} L^2}{2 b^{k-j-1}   \prod_{s=0}^{j-1} \left(\sigma a_{1} T_{1} \left( \frac{c}{b} \right)^{k-s-1}  \right)} \notag \\
& +  \frac{L^2 a_1}{2 b^{k-1}} \left[ 12 \frac{\log{(T_1 c^{k-1} \sqrt{n})}}{1 - \sqrt{\lambda_2}} + 19\right] \notag \\
& + \frac{3L^2 \sum_{s=1}^{k-1} a_1 T_1 \left( \frac{c}{b}\right)^{s-1}  }{T_1 c^{k-1}}.
\end{align}
To ensure convergence to zero, we need $c \geq b$ and  $\sigma a_1 T_1 > 1$ or $a_1 > \frac{1}{T_1 \sigma}$. Given these restrictions, let us make the choices
\begin{align}
a_1 = 1,\ \ T_1 = \left\lceil \frac{2 }{\sigma} \right\rceil ,\ \  c = b = 2.
\end{align}
To simplify the exposition, let us assume that $T_1 = \frac{2}{\sigma}$ is an integer. Using the selected values, we obtain
\begin{align}
F(\hat{w}_i^{k+1}) & - F(w^*)  \leq  \frac{2 L^2}{\sigma \prod_{j=0}^{k-1} \left( 2  \left( \frac{2}{2} \right)^{k-j-1} \right)} \notag \\
& + \sum_{j=1}^{k-1} \frac{L^2}{ \cdot 2 \cdot  2^{k-j-1}   \prod_{s=0}^{j-1} \left(2   \left( \frac{2}{2} \right)^{k-s-1}  \right)} \notag \\
& +  \frac{L^2 }{2 \cdot 2^{k-1}} \left[ 12 \frac{\log{(\frac{2 }{\sigma} \cdot  2^{k-1} \sqrt{n})}}{1 - \sqrt{\lambda_2}} + 19\right] \notag \\
& + \frac{3L^2    \sum_{s=1}^{k-1}  \left( \frac{2}{2}\right)^{s-1}  }{ 2^{k-1}} \\
\leq  & \frac{2 L^2}{\sigma 2^k  } + \sum_{j=1}^{k-1} \frac{L^2}{ 2^{k-j}  2^j  } \notag \\
& +  \frac{L^2 }{2^{k}} \left[ 12 \frac{\log{( \frac{ 2^{k} \sqrt{n} }{\sigma})}}{1 - \sqrt{\lambda_2}} + 19\right]  + \frac{3L^2   (k-1) }{ 2^{k-1}} \\
\leq  & \frac{2 L^2}{\sigma 2^k } + \frac{L^2 (k-1)}{ 2^{k} }\notag \\
& +  \frac{L^2 }{ 2^{k}} \left[ 12 \frac{\log{(\frac{ 2^{k} \sqrt{n} }{\sigma})}}{1 - \sqrt{\lambda_2}} + 19\right]  + \frac{6L^2   (k-1) }{ 2^{k}} \\
\leq  & \frac{2 L^2}{\sigma 2^k } + \frac{L^2 (k-1)}{ 2^{k} }\notag \\
& +  \frac{L^2 }{2^{k}} \left[ 12 \frac{\log{(\frac{ 2^{k} \sqrt{n} }{\sigma})}}{1 - \sqrt{\lambda_2}} + 19\right]  + \frac{6L^2   (k-1) }{ 2^{k}}.
\end{align}
Finally, we have all we need to complete the analysis of Algorithm \ref{alg:dogd}.

\subsection{Proof of Theorem~\ref{thm:dogd_convergence}}

Suppose we run  Algorithm \ref{alg:dogd} for $T$ total steps at each node. This allows for $\tilde{k}$ rounds, where $\tilde{k}$ is determined by solving
\begin{align}
\sum_{i=1}^{\tilde{k}} T_i \leq T \Longleftrightarrow \sum_{i=1}^{\tilde{k}} 2 \cdot 2^i \leq T \Longleftrightarrow \tilde{k} \leq \log_2{\left( \frac{T}{2} + 1\right)}.
\end{align}
Using this value for $k$ we see that
\begin{align}
F(\hat{w}_i^{\tilde{k}+1}) & - F(w^*) \leq   \frac{L^2}{\sigma } 2^{\tilde{k}}  + \frac{L^2 (\tilde{k}-1)}{ 2^{\tilde{k}}} \notag \\
& +  \frac{L^2 }{ 2^{\tilde{k}}} \left[ 12 \frac{\log{( \frac{ 2^{\tilde{k}}  \sqrt{n} }{\sigma})}}{1 - \sqrt{\lambda_2}} + 19\right]  + \frac{6L^2   (\tilde{k}-1) }{  2^{\tilde{k}}} \\
\leq & \frac{L^2}{\sigma  \left( \frac{T}{2} + 1\right) }    + \frac{L^2 (\log_2 \left( \frac{T}{2} + 1\right)-1)}{\left( \frac{T}{2} + 1\right) } \notag \\
& +  \frac{ L^2 }{\left( \frac{T}{2} + 1\right)} \left[ 12 \frac{\log{\left( \frac{ \left( \frac{T}{2} + 1\right) \sqrt{n} }{\sigma}\right)}}{1 - \sqrt{\lambda_2}} + 19\right]  \notag \\
& + \frac{6L^2   (\left( \frac{T}{2} + 1\right)-1) }{\left( \frac{T}{2} + 1\right)} \notag \\
= & O\left(  \frac{ \log{(\sqrt{n} T)}}{ T (1 - \sqrt{\lambda_2})}\right) =  O\left( \frac{\log(\sqrt{n} T)}{T} \right),
\end{align}
when $\lambda_2$ is constant and does not scale with $n$, and this concludes the proof of Theorem~\ref{thm:dogd_convergence}.

\section{Extension to Stochastic Optimization} 
\label{sec:stochastic_opt}
The proof presented in the previous section can easily be extended to the case where each node receives a random estimate $\hat{g}(t)$ of the gradient, satisfying $\mathds{E}[\hat{g}(t)] = g(t)$, instead of receiving $g(t)$ directly. We assume that noisy gradients still have bounded variance i.e., $\mathds{E}[ \norm{\hat{g}_i(t)}^2 ] \leq L^2$. In this setting, instead of equation \eqref{eq:bound_starting_point}, we have
\begin{align}
A_2 = & \sum_{t=1}^{T_k} \frac{1}{n}  \sum_{i=1}^n \langle g_i(t), \overline{w}^k(t) - w^*\rangle \\
= & \sum_{t=1}^{T_k}  \langle \frac{1}{n}  \sum_{i=1}^n \hat{g}_i(t), \overline{w}^k(t) - w^*\rangle \notag \\
& + \sum_{t=1}^{T_k} \frac{1}{n}  \sum_{i=1}^n \langle g_i(t) - \hat{g}_i(t), \overline{w}^k(t) - w^*\rangle.
\end{align}
However, the proof of Theorem~\ref{thm:zink} does not depend on the gradients being correct; rather, it holds for noisy gradients $\hat{g}_(t)$  as well. Moreover, we have $\mathds{E}[\norm{ \hat{g}_i(t)}] \leq L$, and by H\"{o}lder's inequality $\mathds{E}[ \norm{\hat{g}_i(t)} \norm{\hat{g}_j(t)} ] \leq  L^2$. Thus,
\begin{align}
\mathds{E} \left[ \norm{ \frac{1}{n} \sum_{i=1}^n \hat{g}_i(t) } \right] \leq \frac{1}{n^2} \sum_{i,j = 1}^n \mathds{E}[ \norm{\hat{g}_i(t)} \norm{\hat{g}_j(t)} ] \leq L^2.
\end{align}
Thus, invoking Theorem \ref{thm:zink},  if the new data and thus the subgradients are independent of the past, and since $\mathds{E}[\hat{g}_i(t)] = g_i(t)$, we have
\begin{align}
\mathds{E}[ A_2]  \leq & \frac{\norm{\overline{w}^k(1) - w^*}^2}{2 a_k} + \frac{{T_k} a_k L^2 }{2} \notag \\
& +\mathds{E}[\sum_{t=1}^{T_k} \frac{1}{n}  \sum_{i=1}^n \langle g_i(t) -  \hat{g}_i(t), \overline{w}^k(t) - w^*\rangle ]  \\
= & \frac{\norm{\overline{w}^k(1) - w^*}^2}{2 a_k} + \frac{{T_k} a_k L^2 }{2} \notag \\
&+ \sum_{t=1}^{T_k} \frac{1}{n}  \sum_{i=1}^n \langle \mathds{E} \left[ g_i(t) - \hat{g}_i(t) \right], \overline{w}^k(t) - w^*\rangle  \\
 = &\frac{\norm{\overline{w}^k(1) - w^*}^2}{2 a_k} + \frac{{T_k} a_k L^2 }{2}.
\end{align}
Furthermore, the network error bound holds in expectation as well, i.e.,
\begin{align}
\mathds{E}[\norm{\overline{w}^k(t) - w_i^k(t)}& ] \leq  \mathds{E}[\norm{\overline{z}^k(t) - z_i^k(t)}] \notag \\
\leq & 2 a_k L \frac{\log{(T_k \sqrt{n})}}{1 - \sqrt{\lambda_2}} + 3 a_k L  + \frac{L \sum_{s=1}^{k-1} a_s T_s}{T_k} 
\end{align}
Collecting all these observations we have shown that, in expectation,
\begin{align}
\mathds{E}\left[ F(\hat{w}_i^{k+1}) - F(w^*) \right]  \leq &\frac{\norm{\overline{w}^k(1) - w^*}^2}{2 a_k T_k} + \frac{ a_k L^2 }{2} \notag \\
& + L^2 a_k \left[ 6 \frac{\log{(T_k \sqrt{n})}}{1 - \sqrt{\lambda_2}} + 9 \right] \notag \\
& + \frac{3L^2  \sum_{s=1}^{k-1} a_s T_s}{T_k}
\end{align}
which, after using the update rules for $a_k$ and $T_k$, is exactly the same rate as before. We note however that there may still be room for improvement in the distributed stochastic optimization setting since \cite{StochOptStrognlyConvex} describes a single-processor algorithm that converges at a rate $O\left( \frac{1}{T} \right)$.

\section{Simulation} \label{sec:experiments}

To illustrate the performance of DOGD we simulate online training of a classifier by solving the problem \eqref{eq:svm} using a network of $10$ nodes arranged as a random geometric graph. Each node is given $T=600$ data points, and the input dimension is $d=100$. We set $\sigma = 0.1$ and generate the data from a standard normal distribution and classify them as $-1$ or $1$ depending on their relative position to a randomly drawn hyperplane in $\mathds{R}^d$. As we see in Figure \ref{fig:svml2_minimization}, DODG minimizes the objective much faster than Distributed Dual Averaging (DDA)~\cite{dualAveraging} which has a convergence rate of $O\left( \frac{\log(T \sqrt{n})}{\sqrt{T}}\right)$. DDA is simulated using the learning rate that is suggested in~\cite{dualAveraging}. We have observed that boosting this learning rate may yield faster convergence, but still not as fast as DOGD. Figure~\ref{fig:svml2_minimization} also shows the performance of a version of Fast Distributed Gradient Descent (FDGD)~\cite{FastDistributedGradMethods}. As we can see, FDGD fails to converge in an online or stochastic setting and ends up oscillating. 

\begin{figure}
\begin{center}
\includegraphics[width=3.3in]{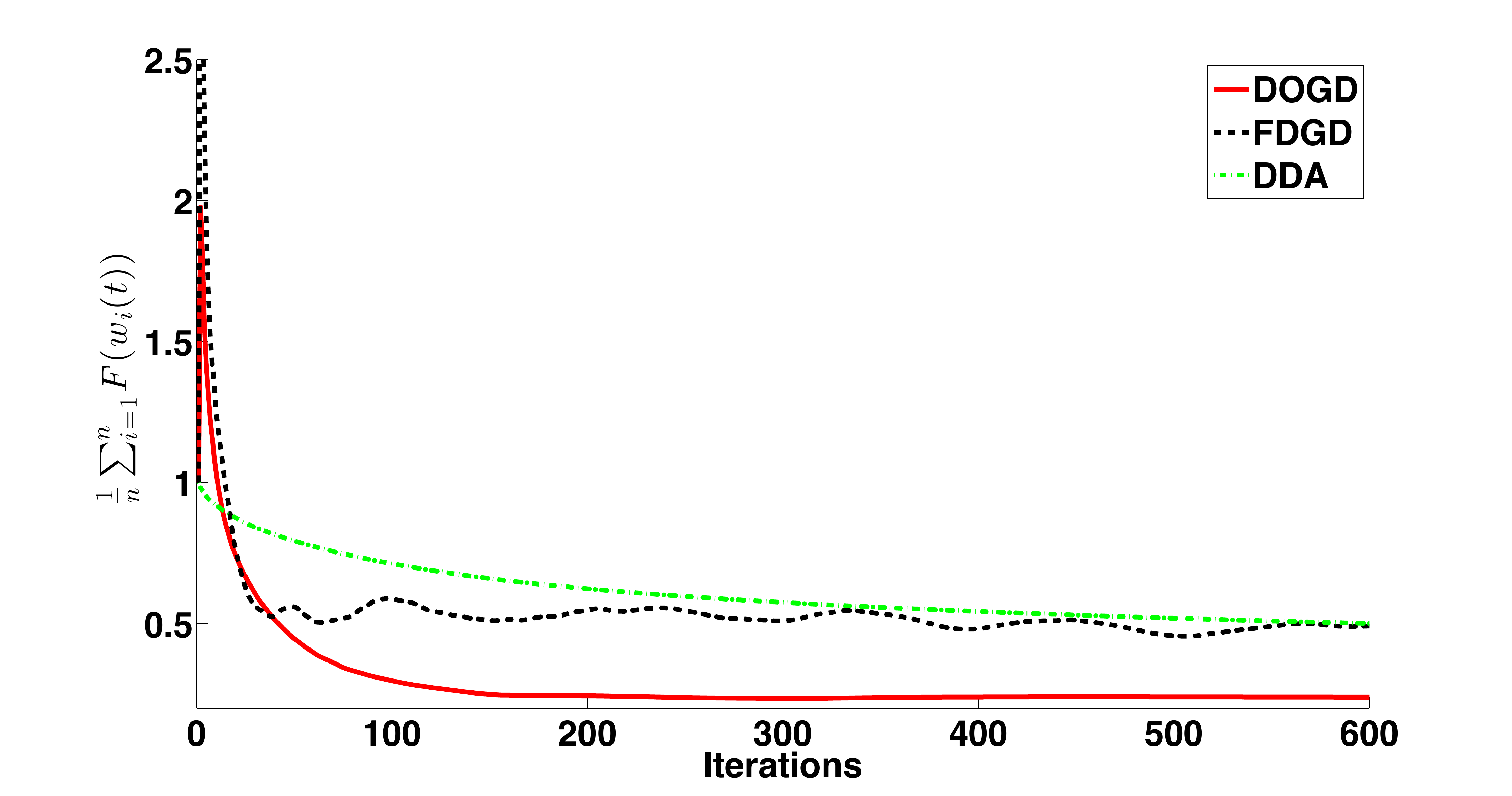} 
\end{center}
\caption{\label{fig:svml2_minimization} Optimization of a $d=100$ dimensional problem of the form \eqref{eq:svm} with a random network of $10$ nodes. Our proposed algorithm DOGD(red) converges faster than DDA(green) as expected from the $T$ instead of $\sqrt{T}$ in the denominator of the convergence rate bound. FDGD(black), is unable to converge in the online problem.}
 \end{figure}

\section{Future Work} \label{sec:future}

In this paper we have proposed and analyzed a novel distributed optimization algorithm which we call Distributed Online Gradient Descent (DOGD). Our analysis shows that DOGD converges at a rate $O(\frac{\log(\sqrt{n}T)}{T})$ when solving online, stochastic or batch constrained convex optimization problems if the objective function is strongly convex. This rate is optimal in the number of iterations for the online and batch setting and slower than a serial algorithm only by a logarithmic factor  in the stochastic optimization setting. 

In its current form, DOGD requires the nodes in the network to exchange gradient information at every iteration. Our preliminary investigation suggests that gradually performing more and more updates between each communication can speed up distributed optimization algorithms in the batch setting when one explicitly accounts for the time required to communicate data. Our future work will carry out a similar analysis for online and stochastic optimization algorithms.

\section*{Appendix}

\begin{lemma} \label{lem:P_tv_convergence} If $P$ is a doubly stochastic matrix defined over a strongly connected graph $G=(V,E)$ with $\abs{V}=n$ nodes so that $p_{ji} = 0$ if $(i,j) \not \in E$, then for any $t \leq  T$,
\begin{align}
\sum_{s=1}^{t-1 }\norm{\frac{1}{n} \vc{1}^T - \left[ P^{t-s+1} \right]_{i,:}}_1 \leq 1 + \frac{\log{(T \sqrt{n})}}{1 - \sqrt{\lambda_2}}
\end{align}
where $\lambda_2$ is the second largest eigenvalue of $P$.
\end{lemma}

\begin{IEEEproof}
If the consensus matrix $P$ is doubly stochastic it is straightforward to show that $P^t \rightarrow \frac{1}{n} \vc{1} \vc{1}^T$ as $t \rightarrow \infty$. Moreover, from standard Perron-Frobenius is it easy to show (see e.g., \cite{StookDiaconis})
\begin{align}
\norm{\frac{1}{n} \vc{1}^T - \left[ P^t \right]_{i,:}}_1 = 2 \norm{\frac{1}{n} \vc{1}^T - \left[ P^t \right]_{i,:}}_{TV} \leq \sqrt{n} \left(\sqrt{\lambda_2}\right)^t
\end{align}
so in our case $\norm{\frac{1}{n} \vc{1}^T - \left[ P^{t - s+1} \right]_{i,:}}_1  \leq  \sqrt{n} \left(\sqrt{\lambda_2}\right)^{t -s + 1}$. Next, demand that the right hand side bound is less than $\sqrt{n} \delta$ with $\delta$ to be determined:
\begin{align}
\sqrt{n} \left(\sqrt{\lambda_2}\right)^{t - s + 1} \leq \sqrt{n} \delta \Rightarrow t - s +1 \geq \frac{\log{(\delta^{-1})}}{\log{(\sqrt{\lambda_2}^{-1})}}.
\end{align}
So with the  choice $\delta^{-1} = \sqrt{n} T$,
\begin{align}
 \norm{\frac{1}{n} \vc{1}^T - \left[ P^{t - s + 1} \right]_{i,:}}_1 \leq \sqrt{n} \frac{1}{\sqrt{n} T} = \frac{1}{T}
\end{align}
if $t - s +1 \geq \frac{\log{(\delta^{-1})}}{\log{(\sqrt{\lambda_2}^{-1})}} = \hat{t}$. When $s$ is large and $t - s + 1 < \hat{t}$ we  take $\norm{\frac{1}{n} \vc{1}^T - \left[ P^{t - s +1} \right]_{i,:}}_1 \leq 2$. The desired bound is not obtained as follows
\begin{align}
\sum_{s=1}^{t-1} \norm{  \frac{1}{n} \vc{1}^T - \left[ P^{ t- s + 1} \right]_{i,:} }_1  = & \sum_{s=1}^{t- \hat{t} -1 } \norm{  \frac{1}{n} \vc{1}^T - \left[ P^{ t- s +1} \right]_{i,:} }_1 \\
& + \sum_{s = t - \hat{t}}^{t - 1} \norm{  \frac{1}{n} \vc{1}^T - \left[ P^{t- s+1} \right]_{i,:} }_1 \notag \\
 \leq &\sum_{s=1}^{t - \hat{t} -1 } \frac{1}{T} + \sum_{s = t - \hat{t}}^{t - 1} 2 \\
\leq & \frac{t - \hat{t} }{T}    + 2 \hat{t} \leq  1  + 2 \hat{t}
\end{align}
Since $t \leq  T$ we know that $t - \hat{t} < T$. Moreover, $\log{(\sqrt{\lambda_2})^{-1}} \geq 1 - \sqrt{\lambda_2}$. Using there two fact we arrive at the result. The same bound is true for any individual entry of $P^t$ approaching  $\frac{1}{n}$.
\end{IEEEproof}

\bibliographystyle{IEEEtran} 
\bibliography{../PhDThesis/References}

\end{document}